\documentclass[11pt, reqno]{amsart} %reqno forces equation numbering on the right
 \usepackage[margin=1in]{geometry}

\usepackage{amssymb, amsmath, amscd, amsthm, color, epsfig, url, mathtools,graphicx}
\usepackage{wasysym}
\usepackage{bm}
\usepackage[font={footnotesize}]{caption}

\usepackage{enumitem}

\makeatletter
\renewcommand\l@subsection{\@tocline{2}{0pt}{2pc}{5pc}{}}
\makeatother

\newcommand{\R}{{\mathbb R}}

\newcommand{\curl}{\operatorname{curl}}
\newcommand{\divr}{\operatorname{div}}
\newcommand{\Diff}{\operatorname{Diff}}
\newcommand{\lk}{\operatorname{lk}}
\newcommand{\Hel}{\mathcal{H}}

\newcommand{\bB}{\mathbf{B}}
\newcommand{\bA}{\mathbf{A}}

\newcommand{\bL}{\mathsf{L}}
\newcommand{\bW}{\mathsf{W}}
\newcommand{\bM}{\mathsf{M}}
\newcommand{\bH}{\mathsf{H}}

\newcommand{\bC}{\mathsf{C}}

\newcommand{\no}{\noindent}

\newtheorem*{rep@theorem}{\rep@title}
\newcommand{\newreptheorem}[2]{%
	\newenvironment{rep#1}[1]{%
		\def\rep@title{#2 \ref{##1}}%
		\begin{rep@theorem}}%
		{\end{rep@theorem}}}

\theoremstyle{plain}

\newreptheorem{thm}{Theorem}
\newtheorem{thmp}{Theorem}[section] 

\newreptheorem{thmp}{Theorem} %do we need this?

\theoremstyle{definition}

\newtheorem{def/ex}[thm]{Definition/Example}

\newtheorem{ques}{Question}

\newtheorem{rems}{Remark}

\graphicspath{{pict/},{./../pict/},{./../../pict/}}

\begin{document}

%%%%%%%%%%%%%%%%%%%%%%%%%%%%%%%%%%%%%%%%%%%%%%

\title[ Woltjer's force free minimizers and Moffatt's magnetic relaxation]{On Woltjer's force free  minimizers and Moffatt's \\ magnetic relaxation}

\date{\today}

%%%%%%%%%%%%%%%%%%%%%%%%%%%%%%%%%%%%%%%%%%%%%%

\author{R. Komendarczyk}

\dedicatory{Dedicated to the memory of S{\l}awek Kwasik (1953--2021)}

\address{Department of Mathematics, Tulane University, 6823 St. Charles Ave, New Orleans, LA 70118}
\email{rako@tulane.edu}
\urladdr{dauns01.math.tulane.edu/\textasciitilde rako}

%\keywords{Braids, loop spaces, cobar construction, configuration space integrals, iterated integrals, formality, graph complexes, trivalent graphs, chord diagrams}

%%%%%%%%%%%%%%%%%%%%%%%%%%%%%%%%%%%%%%%%%%%%%%%%%%%%%%%%%%%%%%%%%%%%%%%%%%%%%%%%%%%%%%%%%%%%%%%%%%%%%%%%

\begin{abstract}
In this note, we exhibit a situation where a stationary state of Moffatt's ideal magnetic relaxation problem is different than the corresponding force-free $L^2$ energy minimizer of Woltjer's variational principle. Such examples have been envisioned in Moffatt's seminal work on the subject and involve divergence free vector fields supported on collections of essentially linked magnetic tubes. Justification of Moffatt's examples requires the strong convergence of a minimizing sequence.  What is proven in the current note is that there is a gap between the global minimum ({\em Woltjer's minimizer}) and the minimum over the weak $L^2$ closure of the class of vector fields obtained from a topologically non-trivial field by energy-decreasing diffeomorphisms. In the context of Taylor's conjecture, our result shows that the Woltjer's minimizer cannot be reached during the viscous MHD relaxation in the perfectly conducting magneto-fluid if the initial field has a nontrivial topology. The result also applies beyond Moffatt's relaxation to any other relaxation process which evolves a divergence free field by means of energy-decreasing diffeomorphisms, such processes were proposed by Vallis et.al and more recently by Nishiyama.
\end{abstract}

%%%%%%%%%%%%%%%%%%%%%%%%%%%%%%%%%%%%%%%%%%%%%%%%%%%%%%%%%%%%%%%%%%%%%%%%%%%%%%%%%%%%%%%%%%%%%%%%%%%%%%%%

\maketitle

%\tableofcontents
%\baselineskip=13pt
%\parskip=4pt
%\parindent=0cm

%%%%%%%%%%%%%%%%%%%%%%%%%%%%%%%%%%%%%%%%%%%%%%%%%%%%%%%%%%%%%%%%%%%%%%%%%%%%%%%%%%%%%%%%%%%%%%%%%%%%%%%%%%%
\section{Introduction}\label{sec:intro}
The Woltjer's variational problem \cite{Woltjer:1958}, known in the context of hydrodynamics and magnetohydrodynamics \cite{Arnold:1986, Arnold:1998, Berger:1984,Cantarella:2000,Boulmezaoud:2000,Enciso:2012, Taylor:1986,Moffatt:1990, Moffatt:2015, Maggioni:2009, Russell:2015, Oberti-Ricca:2018}, concerns the minimization of the $L^2$--energy $E(\mathbf{B})=\int_\Omega \|\mathbf{B}(\mathbf{x})\|^2 d\mathbf{x}$ of over the subspace of divergence free vector fields defined on a regular domain $\Omega$ subject to a helicity constraint. Various boundary conditions, depending on the topology of $\Omega$, can be imposed,  we refer to \cite{Laurence:1991} for further details, here we consider the case of a simply connected domain $\Omega$ with smooth connected boundary. %\footnote{Woltjer's approach used an additional restrictive assumption which has been later removed in \cite{Laurence:1991}.}. 
The formal analysis, presented in \cite{Laurence:1991}, begins with the space 
\begin{equation}\label{eq:W(Omega)}
\bL^2_{\curl}(\Omega)=\{\mathbf{B}\in \bL^2(\Omega)\ |\ \mathbf{B}=\curl(\mathbf{A}),\ \divr(\mathbf{B})=0,\ \mathbf{n}\cdot\mathbf{B} =0\ \text{on}\ \partial \Omega,\ \mathbf{A}\in \bL^2(\Omega)\},
\end{equation}
(where the derivatives are understood in the weak sense, \cite{Gilbarg:2001}) and seeks minimizers of $E(\mathbf{B})$ subject to the constraint:
\[
\Hel(\mathbf{B})=\int_{\Omega} \mathbf{B}\cdot\mathbf{A}\, d\mathbf{x}=(\mathbf{B},\mathbf{A})_{\bL^2}=c_\ast,\qquad \mathbf{B}=\curl(\mathbf{A}),\qquad c_\ast=\text{const}.
\]
The quantity $\Hel(\mathbf{B})$ is called the {\em helicity} of the field $\mathbf{B}$, \cite{Arnold:1986, Arnold:1998, Moffatt:1985, Woltjer:1958} and is an  invariant of $\mathbf{B}$, under the volume preserving deformations i.e. $\Hel(\mathbf{B})=\Hel(f_\ast\mathbf{B})$ for any $f\in \Diff_{0}(\Omega,d\mathbf{x})$ of class $C^1$ (i.e. volume preserving diffeomorphisms, which are equal to the identity along the boundary of $\Omega$). For further reference we state the Woltjer's problem, as follows
\begin{equation}\label{eq:woltjer-var-problem}
\text{minimize $E(\mathbf{B})$  over }\ \bW_{c_\ast}(\Omega)=\{\mathbf{B}\in \bL^2_{\curl}(\Omega)\ |\ \Hel(\mathbf{B})=c_\ast\}.
\end{equation}
\no As shown in \cite{Laurence:1991},  the minimizer $\overline{\mathbf{B}}$ exists and satisfies
$\curl(\overline{\mathbf{B}})=\lambda\overline{\mathbf{B}}$, $\lambda\in \mathbb{R}$,
i.e. $\overline{\mathbf{B}}$ is an eigenfield of the operator $\curl$, and therefore a smooth classical solution of the Euler equations: $\mathbf{B}\cdot\nabla \mathbf{B}=-\nabla P$, where $\nabla\cdot \mathbf{B}=0$, and $P=-\tfrac 12 \|\mathbf{B}\|^2+\operatorname{const}$.

In \cite{Moffatt:1985},  Moffatt considered the following evolution equations
of a viscous and perfectly conductive magneto-fluid
%%%%%%%%%%%%%%%%%%%%%%%%%
\begin{align}\label{eq:MHD-equation-1}
& \rho\bigl(\partial_t \mathbf{v} +\mathbf{v}\cdot\nabla \mathbf{v}\bigr)  = -\nabla p+\curl(\mathbf{B})\times \mathbf{B} +\mu \nabla^2 \mathbf{v},\\
& \label{eq:MHD-equation-2} \partial_t \mathbf{B} =\curl(\mathbf{v}\times \mathbf{B}),\quad \nabla\cdot \mathbf{v}  =\nabla\cdot \mathbf{B} = 0,\\
& \label{eq:MHD-equation-3} \mathbf{B}(\mathbf{x},0)= \mathbf{B}_0(\mathbf{x}),\quad \mathbf{v}(\mathbf{x},0)=0,\\
& \label{eq:bdry-conditions} \mathbf{n}\cdot\mathbf{B} =0,\quad \mathbf{v} = 0,\quad \text{on}\ \partial \Omega,
\end{align}
%%%%%%%%%%%%%%%%%%%%%%%%%
where $\rho$ is the fluid density (assumed uniform), $\mu$ viscosity (in \cite{Moffatt:1985} it is assumed sufficiently large when compared with Reynolds number associated with the flow) $p(\mathbf{x},t)$ is the pressure field.
The second equation in \eqref{eq:MHD-equation-2} assures that $\mathbf{B}_t=\mathbf{B}(\mathbf{x},t)$ is transported with the flow $\phi_{\mathbf{v}}$ of $\mathbf{v}$, i.e.
\begin{equation}\label{eq:frozen-in}
\mathbf{B}(\mathbf{x},t)=\phi_{\mathbf{v}}(\mathbf{x},t)_\ast \mathbf{B}_0(\phi_{\mathbf{v}}(\mathbf{x},-t)),
\end{equation}
where $\phi_{\mathbf{v}}(\mathbf{x},t)_\ast$ denotes a pushforward of the field under the diffeomorhism $\phi_{\mathbf{v}}$. Moffatt further shows that  as long as $\mathbf{v}\neq 0$, the $L^2$ energy of $\mathbf{v}$ and $\mathbf{B}$ decreases as $t\to \infty$, by means of the following formula
\begin{equation}\label{eq:energy-derivative}
 \frac{d}{dt}\bigl(E(\mathbf{B}_t)+\int_\Omega \rho \|\mathbf{v}(\mathbf{x},t)\|^2 d\mathbf{x}\bigr)=-2 \int_\Omega \mu \|\curl(\mathbf{v}(\mathbf{x},t))\|^2 d\mathbf{x}.
\end{equation}
Further, he asserts that a minimizing sequence $\mathbf{B}_t$, $t\to\infty$  (c.f. \cite{Childress:2004, Nishiyama:2002}) should yield a stationary state $\bB=\bB_\infty$ satisfying the Euler's equations.  In \cite{Nishiyama:2002}, Nishiyama observes that a rigorous justification of convergence of $\mathbf{B}_t$ to the stationary state is problematic due to the perfect conductivity of the magnetofluid, and introduces, guided by Vallis et.al, \cite{Vallis:1989}, an alternative to \eqref{eq:MHD-equation-1}--\eqref{eq:bdry-conditions} system which admits a measure-valued solution in the sense of  DiPerna and Majda \cite{DiPerna:1987}.

Since the relaxation of the field $\bB_0$ according to  $\eqref{eq:MHD-equation-1}$--$\eqref{eq:bdry-conditions}$ decreases its energy, a general question studied by Moffatt in \cite{Moffatt:1985} is as follows
\begin{ques}\label{q:main_quest}
	Is a stationary state $\bB_\infty$ of the problem $\eqref{eq:MHD-equation-1}$--$\eqref{eq:bdry-conditions}$ (provided it exists)  the same as the corresponding minimizer in the Woltjer's variational problem \eqref{eq:woltjer-var-problem}?
\end{ques}

\begin{figure}[ht]
	\centering
	\begin{picture}(280,140)% width and height of the picture
	\put(0,0){\includegraphics[width=0.6\linewidth]{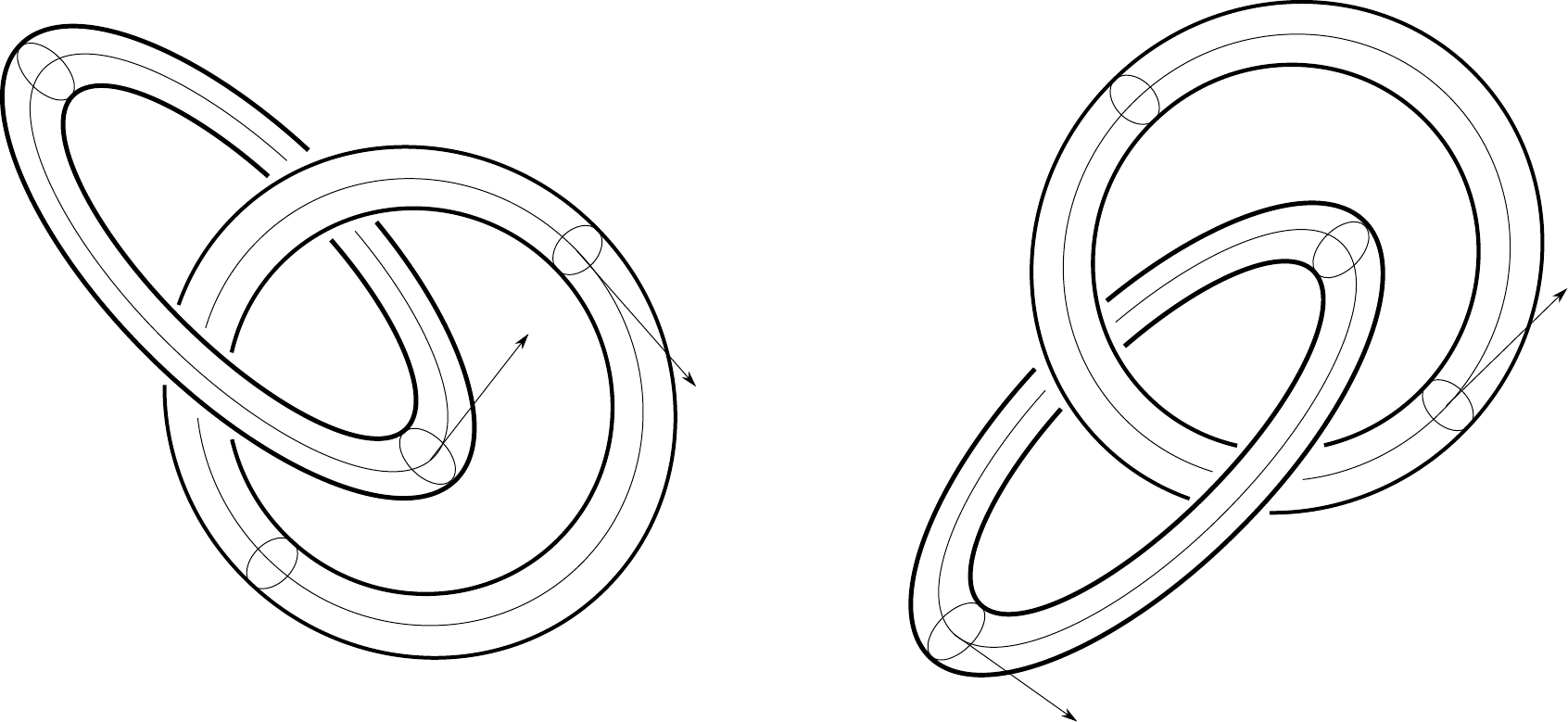}}
	\put(10,50){$\bB^+_h$} 
	\put(15,130){$\mathcal{T}^+_1$}
	\put(105,95){$\mathcal{T}^+_2$}
	
	\put(275,50){$\bB^-_h$}
	\put(145,10){$\mathcal{T}^-_1$}
    \put(240,135){$\mathcal{T}^-_2$}
	\end{picture}
	\caption{Vector field $\bB_h=\bB^+_h+\bB^-_h$ modeled on two Hopf links: $L^+_h$, $L^-_h$ with opposite linking numbers has zero helicity, supported on the tubes $\{\mathcal{T}^\pm_1,\mathcal{T}^\pm_2\}$.}\label{fig:Bh}
\end{figure}
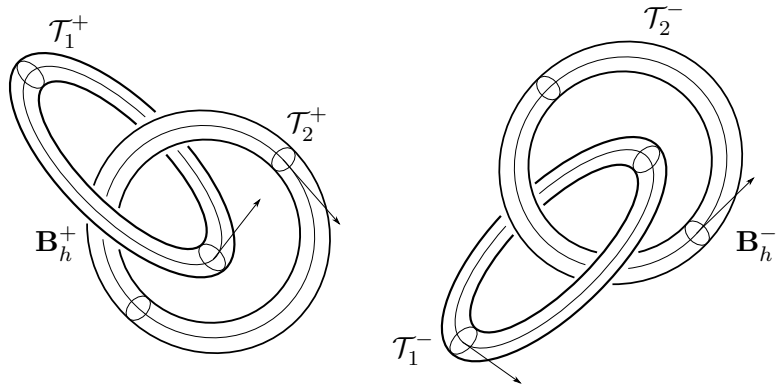
As illustrated in \cite{Moffatt:1985}, one expects the minimizers to be different. In particular, in the case of fields with zero helicity, the force free  minimizer of \eqref{eq:woltjer-var-problem} is the zero field, however a nontrivial topology of the initial field $\bB_0$ can be still prevent a complete energy relaxation. The easiest examples where this situation occurs are the vector fields modeled on essential links and knots in $\R^3$ (see Appendix), Figure \ref{fig:Bh} shows an example of a field modeled on the pair of Hopf links. In \cite{Moffatt:1985}, among other examples, Moffatt considers the field modeled on Borromean rings $\bB_{Borr}$ and observes that the energy $E(\bB_{Borr})$ cannot be decreased to zero under \eqref{eq:MHD-equation-2}, \eqref{eq:frozen-in} thanks to the lower bound of Freedman and He in \cite{Freedman:1991}  (see further remarks after the proof of Theorem \ref{thm:main}). These considerations however require a strong $L^2$ convergence of a minimizing sequence $\mathbf{B}_t$ obtained from \eqref{eq:energy-derivative}, for $t\to \infty$, which is in general problematic as discussed above.
 
In order to circumvent these problems, we observe that any minimizing sequence $\{\mathbf{B}_t\}$ obtained from \eqref{eq:energy-derivative} always has a weakly convergent subsequence.
Thus we consider a variational problem associated with  \eqref{eq:MHD-equation-1}--\eqref{eq:bdry-conditions} which, asks to minimize the $L^2$--energy; $E$  over the subset
 \begin{equation}\label{eq:bar-M-w(Omega,B_0)}
 \begin{split}
 \overline{\bM}^w(\Omega,\mathbf{B}_0) & =\text{the weak $L^2$ closure of} \\
 & \bM(\Omega,\mathbf{B}_0) =\{\mathbf{B}\ |\ \mathbf{B}=f_\ast\mathbf{B}_0, f\in \Diff_{0}(\Omega,d\mathbf{x}); E(\mathbf{B})\leq E(\mathbf{B}_0)\}\subset \bL^2_{\curl}(\Omega),
 \end{split}
 \end{equation}
of divergence free fields obtained from $\mathbf{B}_0\in \bL^2_{\curl}(\Omega)$ via pushforwards by volume preserving diffeomorphisms of $\Omega$ which are identity when restricted to $\partial \Omega$ (the $C^1$ diffeomorphisms are denoted by $\Diff_{0}(\Omega,d\mathbf{x})$). This is consistent with \eqref{eq:MHD-equation-2}, since every vector field in $\bM(\Omega,\mathbf{B}_0)$ has the same topology as the initial field $\mathbf{B}_0$, and \eqref{eq:MHD-equation-2} simply defines a path\footnote{for the perfectly  conducting magneto--fluid the long time existence is not known \cite{Nishiyama:2002}.} in $\bM(\Omega,\bB_0)$. A clear difference with Woltjer's problem is that the helicity constraint provides only a ``mild'' restriction on a topology of a field, whereas vector fields in $\bM(\Omega,\bB_0)$ have equivalent topology to the initial $\bB_0$. 
 
 In order to put it in the general setting, recall that a usual variational problem asks to minimize a weakly lower semicontinuous functional $E$ over a weakly compact class of functions $W$, \cite{Ekeland:1999, Laurence-Stredulinsky:2000}. One then considers a minimizing sequence $f_n\in W$ weakly convergent to $f\in W$, then $E(f)\leq \lim\inf E(f_n)$. By the extreme value theorem for the weakly lower semicontinuous functions, $f$ is a minimizer of $E$ over $W$. Note that $\bW_{c_\ast}(\Omega)$ in \eqref{eq:woltjer-var-problem}, is  weakly closed in $\bL^2(\Omega)$ \cite[p. 1237]{Laurence:1991}. Since $\bM(\Omega,\mathbf{B}_0)$ is not weakly closed, we consider its weak $\bL^2$ closure  $\overline{\bM}^w(\Omega,\mathbf{B}_0)$ and ask to  
 \begin{equation}\label{eq:moffatt-var-problem}
 \text{minimize $E(\mathbf{B})$  over }\ \overline{\bM}^w(\Omega,\mathbf{B}_0).
 \end{equation}
This formulation meets the requirements of the usual variational problem.
  Clearly, the caveat of replacing $\bM(\Omega,\mathbf{B}_0)$ by $\overline{\bM}^w(\Omega,\mathbf{B}_0)$ is that the field line topology is no longer preserved and it is possible that
 \[
   \inf_{\mathbf{B}\in\overline{\bM}^w(\Omega,\mathbf{B}_0)} E(\mathbf{B}) < \inf_{\mathbf{B}\in\bM(\Omega,\mathbf{B}_0)} E(\mathbf{B}),
 \]
examples of paths in $\bM(\Omega,\mathbf{B}_0)$ where this possibility is realized were constructed in \cite{Laurence-Stredulinsky:2000}. However, the relaxation \eqref{eq:MHD-equation-1}--\eqref{eq:bdry-conditions}, under the assumption of the long time existence of classical solutions, always takes place in $\bM(\Omega,\mathbf{B}_0)$, therefore the field line topology is preserved for ``all time'' except in the limit where it can be lost. Moffatt \cite{Moffatt:1985} defines fields {\em topologically accessible} from $\mathbf{B}_0$ (i.e. limits of minimizing sequences $\mathbf{B}_t$ obtaind from \eqref{eq:MHD-equation-2}, \eqref{eq:energy-derivative}), here we simply allow such fields to be all possible weak $L^2$ limits of fields from $\bM(\Omega,\mathbf{B}_0)$. Equivalently, one could say that under Moffatt's assumption, topologically accessible fields are in the strong $L^2$ closure of $\bM(\Omega,\mathbf{B}_0)$, whereas in general we may only show that they  are in the weak $L^2$ closure $\overline{\bM}^w(\Omega,\mathbf{B}_0)$ of $\bM(\Omega,\mathbf{B}_0)$. Theorem \ref{thm:main}, presented in this work, shows that the global minimizer ({\em Woltjer's minimizer}) is not in $\overline{\bM}^w(\Omega,\mathbf{B}_0)$ for a field $\mathbf{B}_0$ with a nontrivial linkage of flowlines.

\begin{rems} We also note that in \cite{Freedman:1999} a rotational magnetic field $\bB_Z$ in the round ball $\Omega=B^3\subset\R^3$ ({\em Zeldovich's neutron star}) is considered and a path $\bB_t$ in $\bM(\Omega,\mathbf{B}_Z)$ constructed such that $\bB_t\to 0$, as $t\to\infty$ in $\bL^2(\Omega)$, however  $\bB_t\not\to 0$ in $\bL^\infty(\Omega)$, which demonstrates that the minimizers may be highly irregular (see also \cite{Arnold:1998}).
\end{rems}

Our notation for the function spaces in the next section is as follows: $\bC^\infty_{0,\divr}(\Omega)$ smooth (test) divergence free compactly supported vector fields on $\Omega$,
$\bL^2(\Omega)=\{X\ |\ X \in L^2(\Omega)\times L^2(\Omega)\times L^2(\Omega)\}$, the square integrable vector fields, 
$\bH^1(\Omega)=\{X\ |\ X \in H^1(\Omega)\times H^1(\Omega)\times H^1(\Omega)\}$ the Sobolev space of $L^2$ vector fields with $L^2$ weak derivatives. 

%%%%%%%%%%%%%%%%%%%%%%%%%%%%%%%%%%%%%%%%%%%%%%%%%%%%%%%%%%%%%%%%%%%%%%%%%%%%%%%%%%%%%%%%%%%%%%%%%%%%%%%%%%%
\section{Statement of the result}\label{sec:statement}
 In a simply connected domain $\Omega\subset \R^3$ with smooth connected boundary, let us consider as an initial vector field $\bB_h$ (i.e. $\bB_0=\bB_h$) modeled on a pair of Hopf links $L=L^+_h \cup L^-_h$ (see Appendix \ref{sec:fields-modeled}) as shown in Figure \ref{fig:Bh}  and supported on the tubes $\{\mathcal{T}^\pm_1,\mathcal{T}^\pm_2\}$. It follows from the well known flux helicity formula \cite{Moffatt:1985}, that the total helicity of $\bB_h$ is zero (we review this in the proof below) and therefore the force free minimizer of \eqref{eq:woltjer-var-problem} is zero, \cite{Laurence:1991}. 

\begin{thmp}\label{thm:main}
	For the initial field $\bB_0=\bB_h$, a minimizer of the problem \eqref{eq:moffatt-var-problem}  is a nonzero field in $\bL^2(\Omega)$.
\end{thmp}
\no Before presenting the proof, let us look closer at the construction of the divergence vector field $\bB_h=$ with zero total helicity but nonzero subhelicites, i.e. $\Hel(\bB_h)=0$ and $\Hel(\bB^+_h)=1$, $\Hel(\bB^-_h)=-1$. Figure \ref{fig:Bh} illustrates the field supported on the tubes $\{\mathcal{T}^\pm_1,\mathcal{T}^\pm_2\}$ about a $4$--component link $L_h=\{L^+_h,L^-_h\}$, which is a disjoint union of two Hopf links $L^+_h=(L^{1,+}_h,L^{2,+}_h)$ and $L^-_h=(L^{1,-}_{h},L^{2,-}_h)$ with opposite linking numbers i.e. 
\begin{equation}\label{eq:linking-L_h}
 \lk(L^+_h)=\lk(L^{1,+}_h,L^{2,+}_h)=1,\qquad \text{and}\qquad \lk(L^-_h)=\lk(L^{1,-}_{h},L^{2,-}_h)=-1. 
\end{equation}
 Let $\Omega$ be a simply connected domain in $\R^3$ with smooth boundary $\partial \Omega$. We set  $\mathbf{B}_h$ to be the divergence free vector field modeled on $L_h$ as defined in Appendix \ref{sec:fields-modeled} supported on the disjoint tubes around the link $L_h$. Restricting $\mathbf{B}_h$ to each individual tube we obtain
%%%%%%%%%%%%%%%%%%%%%%%%%%%%%%%%%%
\begin{equation}\label{eq:B-pm-decomposition}
\mathbf{B}_h=\mathbf{B}^+_h+\mathbf{B}^-_h=(\mathbf{B}^{1,+}_h+\mathbf{B}^{2,+}_h)+(\mathbf{B}^{1,-}_h+\mathbf{B}^{2,-}_h).
\end{equation}
We may assume that the fields $\mathbf{B}^+_h$ and $\mathbf{B}^+_h$ and the supporting tubes $\{\mathcal{T}^\pm_1,\mathcal{T}^\pm_2\}$ are isometric images of each other in $\Omega\subset \R^3$ as well as  (the isometry needs to reverse the orientation in one of the tubes to obtain \eqref{eq:linking-L_h}). 

\no The above construction yields the following helicity and cross-helicity identities 
\begin{equation}\label{eq:B_h-helicity}
 \Hel(\mathbf{B}^{1,\pm}_h) =\Hel(\mathbf{B}^{2,\pm}_h)=0,\quad \Hel(\mathbf{B}^{1,\pm}_h,\mathbf{B}^{2,\pm}_h)=\lk(L^\pm_h)\Phi(\mathbf{B}^\pm_{h,1})\Phi(\mathbf{B}^\pm_{h,2}),\quad
 \Hel(\mathbf{B}^{\pm,1}_h,\mathbf{B}^{2,\mp}_h)=0.
\end{equation}

Also, without loss of generality, we may scale the fields to obtain the unit fluxes i.e. $\Phi(\mathbf{B}^\pm_{h,\ast})=1$ and $\Hel(\mathbf{B}^\pm_{h,1},\mathbf{B}^\pm_{h,2})=\pm 1$.
Further, $\Hel$ is a symmetric bilinear, thus the above identities yield,
\begin{equation}\label{eq:B_h-zero-helicity}
 \Hel(\mathbf{B}_h)=\Hel(\mathbf{B}^+_h)+\Hel(\mathbf{B}^-_h)=0.
\end{equation}

\begin{rems}
	Recall that the cross--helicity of two fields $\bB_1$ and $\bB_2$ in $\bL^2_{\curl}(\Omega)$ is defined by  
	\begin{equation}\label{eq:cross_helicity}
	 \Hel(\bB_1,\bB_2)=(\bB_1,\bA_2)_{\bL^2}=\int_{\Omega} \bB_1(\mathbf{x})\cdot\bA_2(\mathbf{x}) d\mathbf{x},\quad \bB_2=\curl(\bA_2)
	\end{equation}
	and is a symmetric bilinear form on $\bL^2_{\curl}(\Omega)$. The single field helicity $\Hel(\bB)$ equals $\Hel(\bB,\bB)$, i.e. the associated quadratic form.
\end{rems}

\begin{proof}[Proof of Theorem \ref{thm:main}]
The set $\bM(\Omega,\mathbf{B}_h)$ is bounded in $L^2$ norm, and therefore $\overline{\bM}^w(\Omega,\mathbf{B}_h)$ is weakly compact. Since $E(\,\cdot\,)$ is weakly lower semicontinuous, the extreme value theorem tells us that a minimizer $\overline{\bB}_h$ of $E$ exists over $\overline{\bM}^w(\Omega,\mathbf{B}_h)$, i.e. 
\[
E(\overline{\bB}_h)=\min_{\bB\in \overline{\bM}^w(\Omega,\bB_h)} E(\bB).
\]
By the Eberlein--Smulian Theorem \cite{Rudin:1991}, there is a sequence $\{\bB_{h,n}\}\subset \bM(\Omega,\bB_h)$, $\bB_{h,n}=f_{n,\ast} \bB_h$, $f_n\in \Diff_{0}(\Omega,d\mathbf{x})$, weakly convergent to $\overline{\bB}_h$. The vector field push--forward $f_\ast$ is linear, so the decomposition \eqref{eq:B-pm-decomposition} hold for every $n$:
\[
 \bB_{h,n}=\bB^+_{h,n}+\bB^-_{h,n}\,,\qquad \bB^\pm_{h,n}=f_{n,\ast} \bB^\pm_{h}\,.
\] 
Since the supports of $\bB^+_{h,n}$ and $\bB^-_{h,n}$ are disjoint 
\begin{equation}\label{eq:B_h-norm-bounded}
 \|\bB^\pm_{h,n}\|_{\bL^2}\leq  \|\bB_{h,n}\|_{\bL^2}\leq \|\bB_h\|_{\bL^2},
\end{equation}
i.e. sequences $\{\bB^+_{h,n}\}$, $\{\bB^-_{h,n}\}$ are bounded and therefore weakly convergent (after passing to a subsequence, if necessary), let $\bB^\pm_{h,n}\xrightharpoonup{\quad} \overline{\bB}^\pm_{h}$ clearly, $\overline{\bB}_{h}=\overline{\bB}^+_{h}+\overline{\bB}^-_{h}$. In the next step, we follow the analysis in \cite[p. 1244]{Laurence:1991}: for each 
$\bB^\pm_{h,n}$, we may choose a potential field $\bA^\pm_{h,n}$ in $\bH^1(\Omega)$, such that
\[
 \curl(\bA^\pm_{h,n})=\bB^\pm_{h,n}, 
\]
in the weak sense (i.e. for any $X\in \bC^\infty_{0,\divr}(\Omega)$:
$(\bA^\pm_{h,n},\curl(X))_{\bL^2}=(\bB^\pm_{h,n},X)_{\bL^2}$). The potential fields can be also chosen to  satisfy 
\begin{equation}\label{eq:A-div-boundary-condition}
 \divr(\bA^\pm_{h,n})=0,\quad \bA^\pm_{h,n}\times\mathbf{n}=0,\ \text{along}\ \partial\Omega,
\end{equation}
where $\mathbf{n}$ is the unit normal along $\partial \Omega$ (these identities are in the weak and trace sense).  By Friedrichs inequality \cite{Friedrichs:1955}, if $\bA$ satisfies conditions in \eqref{eq:A-div-boundary-condition}, then
\[
\|\bA\|^2_{\bH^1(\Omega)}\leq c_1(\Omega) \|\curl(\bA)\|^2_{\bL^2}.
\]
From \eqref{eq:B_h-norm-bounded}, sequences $\{\bA^\pm_{h,n}\}$, $\{\bA_{h,n}\}$ are bounded in $\bH^1(\Omega)$, thus 
the Rellich compactness theorem \cite{Gilbarg:2001} implies the following convergences (after passing to a subsequence if necessary) 
\begin{equation}\label{eq:weak-strong-convergence}
 \begin{split}
 \bA^\pm_{h,n}, \bA_{h,n} & \longrightarrow \overline{\bA}^\pm_h, \overline{\bA}_{h}\qquad \text{strongly in}\ \bL^2,\\
 \bB^\pm_{h,n}, \bB_{h,n} & \xrightharpoonup{\quad} \overline{\bB}^\pm_h, \overline{\bB}_{h}\qquad \text{weakly in}\ \bL^2.
 \end{split}
\end{equation}

 Suppose that, contrary to the statement of Theorem \ref{thm:main}, the minimizer of \eqref{eq:moffatt-var-problem} is the zero field, i.e. $\overline{\bB}_h=0$ in $\bL^2(\Omega)$. Using the Hodge decomposition of \cite[p. 879]{Cantarella:2000}, on the simply connected $\Omega$, for any $X\in \bC^\infty_{0,\divr}(\Omega)$, we have $Y\in\bL^2(\Omega)$, such that $X=\curl(Y)$. By the weak convergence in \eqref{eq:weak-strong-convergence}
 \[
  (\bA_{h,n},X)_{\bL^2}= (\curl(\bA_{h,n}),Y)_{\bL^2}= (\bB_{h,n},Y)_{\bL^2}\longrightarrow 0,\qquad \text{as}\quad n\to\infty,
 \]
 where in the first identity we used\footnote{$\int_{\Omega} \langle \curl(W),V\rangle d\mathbf{x}=\pm \int_{\partial \Omega} \langle W\times V,\mathbf{n}\rangle d\sigma+\int_{\Omega} \langle W,\curl(V)\rangle d\mathbf{x}$} the boundary condition \eqref{eq:A-div-boundary-condition}.
 
 Since the weak limit of $\bA_{h,n}$ is the zero field, the strong limit is also the zero field, i.e. $\overline{\bA}_h=\overline{\bA}^+_h+\overline{\bA}^-_h=0$. From the computations in \eqref{eq:B_h-helicity} and the helicity invariance under $\Diff_{0}(\Omega,d\mathbf{x})$, we obtain
  \[
 \begin{split}
  \Hel(\bB_{h,n}) & =(\bB_{h,n},\bA_{h,n})_{\bL^2}\longrightarrow (\overline{\bB}_{h},\overline{\bA}_{h})_{\bL^2}= \Hel(\overline{\bB}_{h})=0\\
  \Hel(\bB^\pm_{h,n}) & =(\bB^\pm_{h,n},\bA^\pm_{h,n})_{\bL^2}\longrightarrow (\overline{\bB}^\pm_{h},\overline{\bA}^\pm_{h})_{\bL^2}= \Hel(\overline{\bB}^\pm_{h})=\pm 1,
 \end{split}
 \]
 (because the inner product of the strongly convergent and weakly convergent sequences is convergent in $\R$.)
The strong convergence: $\bA_{h,n}\longrightarrow 0$, implies\footnote{analogously for $\bB^-_{h,n}$} $(\bB^+_{h,n},\bA_{h,n})_{\bL^2}\longrightarrow 0$, but  from \eqref{eq:B_h-helicity}
\[
 (\bB^+_{h,n},\bA_{h,n})_{\bL^2}=(\bB^+_{h,n},\bA^+_{h,n})_{\bL^2}+(\bB^+_{h,n},\bA^-_{h,n})_{\bL^2}=\Hel(\bB^+_{h,n})\longrightarrow 1,
 \]
since  $(\bB^+_{h,n},\bA^-_{h,n})_{\bL^2}=0$ for every $n$. Thus a contradiction to the assumption  $\overline{\bB}_h=0$. 
\end{proof}

 If one could assume a strong $L^2$ convergence $\bB^\pm_{h,n}\longrightarrow \overline{\bB}^\pm_{h}$ in the proof of Theorem \ref{thm:main} then the classical energy--helicity estimate \cite{Arnold:1986}:
  $c_1(\Omega)|\Hel(\bB)|\leq E(\bB)$,
 can be used to argue that $\overline{\bB}_{h}\neq 0$. 
 Alternatively, one can use the asymptotic crossing number estimate of \cite{Freedman:1988}. Therefore, these classical energy estimates suffice to show that there is a gap between the global minimum  (the zero field, in this case) and the minimum over the strong $L^2$ closure of $\bM(\Omega,\mathbf{B}_0)$.  
Theorem \ref{thm:main} shows a stronger fact i.e. that there is a gap between the global minimum and the minimum over the weak $L^2$ closure $\overline{\bM}^w(\Omega,\mathbf{B}_0)$. Clearly, this result applies beyond the Moffatt's relaxation to any other relaxation process which evolves a divergence free field by means of energy-decreasing diffeomorphisms, such as Vallis \cite{Vallis:1989} and Nishiyama \cite{Nishiyama:2002}. In particular, Theorem \ref{thm:main} answers positively the question posed in \cite[p. 417]{Nishiyama:2002}.  In the context of Taylor's conjecture \cite{Taylor:1986} it imples that the Woltjer's  minimizer cannot be reached during the perfectly conductive relaxation phase, if the inital field has a nontrivial topology. In the follow--up work we address this problem in the case of vector fields modeled on links exhibiting a higher order linkage (such as Borromean rings).  

%%%%%%%%%%%%%%%%%%%%%%%%%%%%%%%%%%%%%%%%%%%%%%%%%%%%%%%%%%%%%%%%%%%%%%%%%%%%%%%%
\appendix

\section{Vector fields modeled on a link.}\label{sec:fields-modeled}

We begin by reviewing a definition of the divergence free vector field modeled on a link (c.f. \cite{Freedman:1988}).
Recall, that an $n$--component link in $\mathbb{R}^3$ is a smooth embedding\footnote{we often identify a link $L$ with is image in $\R^3$}  
\[
L:\bigsqcup^n_{k=1} S^1_k \longrightarrow \R^3,\quad L_k=L\bigl|_{S^1_k}\qquad n\geq 1,
\]
$L$ is called a {\em trivial link} if each component $L_k$ is a boundary of an embedded disk, and the disks are disjoint from the link $L$ itself, otherwise the link is called {\em nontrivial} or {\em essential}. A divergence free vector field $\mathbf{V}=\mathbf{V}_L$ is said to be {\em modeled} on a link $L$, \cite{Freedman:1988}, whenever there is a smooth volume preserving embedding
\[
e_L: \bigsqcup^n_{k=1} D^2_k\times S^1_k \longrightarrow \R^3,
\] 
of solid tori (tubes) $\mathcal{T}_k=e_L(D^2_k\times S^1_k)$ into $\R^3$ such that $e_L|_{\{0\}\times S^1_k}=L_k$, i.e. the cores of the tubes are mapped to the link $L$. Further $\mathbf{V}_L$ restricted to each $\mathcal{T}_k$ is given by   
\[
\mathbf{V}_L\bigl|_{\mathcal{T}_k}=(e_L)_\ast(\phi_k(\mathbf{x})\frac{\partial}{\partial t}),\qquad \mathcal{T}_k=e_L(D^2_k\times S^1_k),
\]
where $(\mathbf{x},t)$ are coordinates on $D^2_k\times S^1_k$ and $\phi_k:D^2_k\longrightarrow [0,1]$ is a unit mass bump function vanishing in some neighborhood of $\partial D^2_k$. Observe that in each tube $\mathcal{T}_k$ the vector field  $\mathbf{V}_L$ is the pushforward of $X_k(\mathbf{x},t)=\phi_k(\mathbf{x})\frac{\partial}{\partial t}$ and the circular orbits $\{\mathbf{x}\}\times S^1_k$ of $X_k$ are mapped to the circular orbits $\gamma_k(\mathbf{x},t)$ of $\mathbf{V}_L$ in $\mathcal{T}_k$. Extending $\mathbf{V}_L$ by zero to the entire domain we obtain a smooth vector field vanishing at $\partial \mathcal{T}$ ($\mathcal{T}=\bigcup_k \mathcal{T}_k$), such that $\mathbf{V}_L=\sum^n_{k=1}\mathbf{V}_k$, where $\mathbf{V}_k=\mathbf{V}_L\bigl|_{\mathcal{T}_k}$. As observed in \cite{Freedman:1988}, the Moser's result \cite{Moser:1965} can be used to make the embedding $e_L$ volume preserving and thus $\mathbf{V}_L$ a divergence free field (as $X_k$ is itself divergence free). Further, $e_L$ can be chosen such that $\operatorname{lk}(\gamma_k(\mathbf{x},t),\gamma_k(\mathbf{y},t))$, $\mathbf{x}\neq \mathbf{y}$, i.e. the pairwise linking of orbits of $\mathbf{V}_L$ within each tube $\mathcal{T}_k$ is zero, such $\mathbf{V}_L$  then satisfies 
\begin{equation}\label{eq:V_k-lk=zero}
\Hel(\mathbf{V}_k)=0,\qquad \mathbf{V}_k=\mathbf{V}_L\bigl|_{\mathcal{T}_k}.
\end{equation}

\end{document}